\documentclass[english,conference,12pt]{article}

\usepackage[T1]{fontenc}
\usepackage[latin9]{inputenc}
\usepackage[top=2cm,bottom=2cm,left=2cm,right=2cm]{geometry}
\usepackage{amsthm}
\usepackage{amsmath}
\usepackage{amssymb}
\usepackage{multirow}

\makeatletter
\theoremstyle{plain}
\newtheorem{thm}{\protect\theoremname}
\theoremstyle{definition}
\newtheorem{defn}[thm]{\protect\definitionname}
\theoremstyle{plain}

\theoremstyle{plain}
\newtheorem{prop}[thm]{\protect\propositionname}
\theoremstyle{definition}

\theoremstyle{remark}

\usepackage{cite}
\author{Farzad Farnoud (Hassanzadeh), Behrouz Touri, and Olgica Milenkovic\\ 
University of Illinois, Urbana-Champaign, IL\\
{\small E-mail: \{hassanz1, touri1, milenkov\}@illinois.edu}}
\date{}
\makeatother

\usepackage{babel}
\providecommand{\definitionname}{Definition}
\providecommand{\examplename}{Example}
\providecommand{\lemmaname}{Lemma}
\providecommand{\propositionname}{Proposition}
\providecommand{\remarkname}{Remark}
\providecommand{\theoremname}{Theorem}

\begin{document}
\global\long\def\dist{\mathsf{d}}
\global\long\def\kdist{\mathsf{d}_{\tau}}
\global\long\def\Kist{\mathsf{d}_{K}}
\global\long\def\Twist{\mathsf{d}_{T}}
\global\long\def\define{:\,=}
\global\long\def\dash{\mbox{--}}
\global\long\def\adjust{A}
\global\long\def\lineset{L}

\title{\vspace{10pt}Nonuniform Vote Aggregation Algorithms}
\maketitle
\begin{abstract}
We consider the problem of non-uniform vote aggregation, and in particular, the algorithmic aspects
associated with the aggregation process. For a novel class of weighted distance measures on votes,
we present two different aggregation methods. The first algorithm
is based on approximating the weighted distance measure by Spearman's footrule distance, with
provable constant approximation guarantees. The second algorithm is based
on a non-uniform Markov chain method inspired by PageRank, for which currently only heuristic guarantees
are known. We illustrate the performance of the proposed algorithms on a number of distance measures for which
the optimal solution may be easily computed.
 \end{abstract}

\section{Introduction}

Vote (rank) aggregation has a long history, dating back to the first democratic elections held in the
polis of Athens under Solon and Cleisthenis~\cite{Sinclair-democracy-1988}. The early voting process involved ranking
two candidates, so that  the problem of determining the winner reduced to simple plurality vote counts.
In more recent political history, it was recognized that plurality
methods, as well as majority pairwise counts for multiple candidate ranking systems are plagued by a number of issues. These issues were most
succinctly identified by de Borda and Condorcet~\cite{mueller2008public}, and pertain to the fact that votes may be
\emph{non-transitive}, that ``strong candidates'' may loose out to ``weak candidates'' due
to their close mutual competition, and that majority pairwise counts may differ substantially from plurality counts.
%Non-transitivity pertains to the fact that the majority of the voters may
%prefer candidate $A$ to candidate $B$ as well as candidate $B$ to candidate $C$, and still have the
%majority prefer candidate $C$ to candidate $A$. A simple
%example consists of three votes/rankings of three candidates, $(A, B, C), (B, C, A)$ and $(C, A, B)$.
%The Borda count tends to favor candidates supported by a broad consensus among voters, rather than the candidate who is necessarily the favorite of a majority; for this reason, some of its supporters see the Borda count %as a method that promotes consensus and avoids the 'tyranny of the majority'. Advocates argue, for example, that where the majority candidate is strongly opposed by a large minority of the electorate, the Borda winner may %have higher overall utility than the majority winner. On grounds such as these, the de Borda Institute of Northern Ireland advocates the use of a form of referendum based on the Borda count in divided societies such as %Northern Ireland, the Balkans and Kashmir.

The above described issues with aggregating multiple votes/rankings led to a line of work centered around the
use of \emph{distance measures} between rankings~\cite{deza2009encyclopedia}
and an underlying axiomatic approach~\cite{arrow1963social}.
The idea behind the distance-based approach is that one can cast the aggregation problem as one of evaluating the median of
a set of points in a given metric space. Well-known metrics used for computing the median
include Kendall's $\tau$ and Spearman's footrule \cite{diaconis1988group}.

One of the drawbacks of aforementioned distance-based aggregation methods is that the distance functions
do not cater to the need of certain applications where similar items are to be treated similarly in the aggregation
process, and where the top vs. the bottom of the list carry different relevance in the ranking~\cite{kumar2010gdr}.
An example for the first scenario may be in ranking candidates for a number of positions, with the constraint that some candidate diversity criteria are met.
An example for the second scenario may be in ranking candidates where only a small fraction at the top is considered
for a position.

In two companion papers~\cite{farnoud2012sorting,cdc-farnoud}, we studied a very general class of distance measures, termed \emph{cost-constrained
transposition distances}, or \emph{weighted transposition distance}, which, among other things, address the two aggregation issues described above. The crux
of this new approach to measuring distances between rankings is to assign non-uniform swapping costs (weights) to different
pairs of locations in the list, or equivalently, different pairs of elements in the inverse list. Aggregation
methods based on weighted transpositions are currently unknown, and the topic of interest in this paper.

The results we present pertain to algorithmic aspects of the weighted vote/rank aggregation problem~\cite{kemeney1959mathematics,cook1985ordinal,sculley2007rank,schalekamp2009rank}.
We describe three algorithms: a constant-approximation algorithm that uses an analytical bound between the weighted
distance and a generalization of Spearman's footrule,
and then solves a minimum weight matching problem; this algorithm is inspired by a procedure described in~\cite{dwork2001rank,dwork2001rank-web};
an aggregation method reminiscent of PageRank~\cite{dwork2001rank}, where the ``hyperlink probabilities'' are chosen according to the
swapping weights; and a combination of the first algorithm with local descent methods.

The paper is organized as follows. A brief introduction to vote aggregation and the problem formulation are given in Section~\ref{sec:preliminaries}.
The main contribution of the paper is presented in Section~\ref{sec:main-results}, which contains the proposed aggregation algorithms.
Results of various rank aggregation processes on an Academic Climate Study dataset gathered at UIUC are presented in Section~\ref{sec:results}.

\section{Preliminaries}\label{sec:preliminaries}

Suppose that an election process includes $m$ voters, each of which provides a ranking of $n$ candidates. These rankings are
collected in a set $\Sigma=\{\sigma_{1},\sigma_{2},\cdots,\sigma_{m}\}$, where each ranking $\sigma_i$ is represented by
a permutation in $\mathbb{S}_{n}$, the symmetric group of order $n$.

Given a distance function $\dist$ over the permutations in $\mathbb{S}_{n}$,
the distance-based aggregation problem can be stated as follows. Find the ranking $\pi^*$ that minimizes the cumulative distance from $\Sigma$, i.e.,
\begin{equation} \label{eqn:rank-agg}
\pi^*=\arg\min_{\pi\in\mathbb{S}_{n}}\sum_{i=1}^{m}\dist(\pi,\sigma_{i}).
\end{equation}

Clearly, the choice of the distance function $\dist$ is an important feature of the aggregation method.
We describe next a few such distance measures, including Kendall's $\tau$ and Spearman's footrule distance~\cite{kemeney1959mathematics}.

Let $e=12\cdots n$ denote the identity permutation (ranking).
\begin{defn}
A transposition $(a\ b)$ in a permutation $\pi$ is the swap of elements in positions $a$ and $b$. When there is no confusion, we consider
a transposition to be a permutation. If $|a-b|=1$, the transposition
is referred to as an \emph{adjacent transposition.}
\end{defn}

It is well-known that any permutation may be reduced to the identity via transpositions
or adjacent transpositions. The former process is referred to as sorting,
while the later is known as sorting with adjacent transpositions.
The smallest number of  transpositions needed to sort a permutation
$\pi$ is known as the Cayley distance, $T(e,\pi)$, while the smallest number of
adjacent transpositions needed to sort a permutation is known as the Kendall's
$\tau$ distance, $K(e,\pi)$.

Let $\Theta=\{{(a\ b): \; a \neq b, \; a,b \in [n]\}}$ be the set of transpositions, and endow $\Theta$ with a non-negative weight function
$\varphi: \Theta \to \mathbb{R}^{+}$ where $\varphi(a,b)$ is the weight of transposition $(a\ b)$. The distance measure of interest is defined as the minimum weight of a sequence of transpositions needed
to transform one permutation $\pi$ into another permutation $\sigma$.
This distance measure is termed the \emph{weighted transposition distance}, and is denoted by $\dist_{\varphi}(\pi,\sigma)$~\cite{farnoud2012sorting}.
It can easily be shown that most distance measures used for rank aggregation represent special
cases of the weighted transposition distance:
\begin{itemize}
\item Kendall's $\tau$, $K(\pi,\sigma) = \dist_{\varphi_K}(\pi,\sigma)$, where $\varphi_K(i,j) = 1$, for $|i-j|=1$ and $\varphi_K(i,j) = \infty$ otherwise.
\item Spearman's footrule, $F(\pi,\sigma)$, defined as $\sum_{i=1}^n |\pi^{-1}(i)-\sigma^{-1}(i)|$, may be written as $F(\pi,\sigma) = \dist_{\varphi_F}(\pi,\sigma)$, where $\varphi_F(i,j)= |i-j|$.
\item Cayley's distance, $T(\pi,\sigma) = \dist_{\varphi_T}(\pi,\sigma),$ where $\varphi_T(i,j) = 1$ for all $i,j$.
\end{itemize}

In what follows, we focus on the \emph{weighted Kendall distance}, where $\varphi(i,i+1) =\varphi(i+1,i) = w_i$, with $w_i$ being non-negative, and $\varphi(i,j) = \infty$ for $|j-i|\neq 1$.

The weighted Kendall distance between two permutations addresses the issue of the top-versus-bottom problem as follows.
To model the significance of the top of the list versus the bottom of the list, one may choose $w_i=n-i$. This means that the weight of swapping the first and the second rank (location) is $n-1$ while the weight of swapping the $n-1^{\text{st}}$ and the $n^{\text{th}}$ rank (location) is $1$.
In this case, transposition weights decay arithmetically as we move towards the end of the list. We may also choose $w_i=c^i$ for $0 \le c< 1$. In this case, the weight decay
is geometric. Weighted transposition distance, in its general form, can be used to model similarities between elements by assigning small weights to transpositions involving similar elements and large weights to transpositions involving
dissimilar elements. Note that in this case, the weights relate to transposing \emph{elements} and not \emph{ranks}, and thus the distance between two permutations $\pi$ and $\sigma$ is defined as
$\dist_\varphi(\pi^{-1},\sigma^{-1})$.

In the next section, we describe how to perform efficient (approximate) rank aggregation using the weighted Kendall
distance. Our results are inspired by related algorithmic approaches proposed in~\cite{dwork2001rank}.

\section{Algorithmic Results} \label{sec:main-results}

Rank aggregation is a combinatorial optimization problem over the set of permutations, and as such, it is computationally costly.
Aggregation with Kendall's $\tau$ distance is known to be NP-hard \cite{dwork2001rank}. However, assuming that $\pi^*$ is the solution to \eqref{eqn:rank-agg}, the ranking
$\sigma_i$ closest to $\pi^*$ provides a 2-approximation for the rank aggregate. This easily follows from the fact that Kendall's $\tau$
satisfies the triangle inequality. As a result, one only has to evaluate the pairwise distances of the votes $\Sigma$ in order
to identify a constant approximation aggregate for the problem. Although we do not provide a detailed proof of this claim, the same is true of the weighted Kendall  distance.

\subsection{Minimum Weight Bipartite Matching Algorithms}

For any distance function that may be written as \(\dist(\pi,\sigma)=\sum_{k=1}^n f(\pi^{-1}(k),\sigma^{-1}(k)),\) where $f$ denotes an arbitrary non-negative function,
one can find an \emph{exact solution} to \eqref{eqn:rank-agg} as follows \cite{dwork2001rank}. Consider a complete weighted bipartite graph $\mathcal{G}=(X,Y)$, with $X=\{1,2,\cdots,n\}$
corresponding to the $n$ ranks to be filled and $Y=\{1,2,\cdots,n\}$ corresponding to the elements of $[n]$, i.e., the candidates.
Let $(i,j)$ denote an edge between $i \in X$ and $j \in Y$. We say that a perfect bipartite matching $P$ corresponds to a permutation $\pi$ whenever $(i,j)\in P$ if and only if
$\pi(i)=j$.
If the weight of $(i,j)$ equals
\begin{equation} \label{eq:linear}
\sum_{l=1}^m f(i,\sigma^{-1}_l(j)),
\end{equation}
i.e., the weight incurred by $\pi(i)=j$, then the minimum weight perfect matching corresponds to a solution of~\eqref{eqn:rank-agg}.

For example, if $\varphi$ is a metric path weight functions\footnote{A metric path weight is a weight function obtained by arranging the elements in [n] on a straight path and by assigning non-negative weights to the edges of the path. The weight of transposing elements in positions $a$ and $b$, $\varphi(a,b)$, is the weight of the unique path between $a$ and $b$.},
we have
\begin{equation*}\begin{split}
\dist_\varphi(\pi,\sigma)&=\sum_{j=1}^{n}\, \varphi(\pi^{-1}(j),\sigma^{-1}(j))\\
&=\sum_{j=1}^{n}\, f(\pi^{-1}(j),\sigma^{-1}(j))
\end{split}\end{equation*}
 where $f=\varphi$.

 Furthermore, note that $\sum_{k=1}^n f(\pi^{-1}(k),\sigma^{-1}(k))$ is a generalization of Spearman's footrule and thus for Spearman's footrule, one may find the exact solution as well.

Let $\mathcal{H}$ denote a complete, undirected graph with vertex set $V=[n]$. To each edge $(i,j), \, i,j \in [n]$ assign the weight $\varphi(i ,j)$.
Furthermore, let $p^*(i,j)$ denote the shortest path (i.e., minimum
weight path) between $i$ and $j$ in $\mathcal{H}$, and let ${\rm weight}(p^*)$ stand for the weight of the shortest path.
\begin{thm}
%Let $\dist_\varphi(\pi,\sigma)$ be the weight of transforming $\pi$ to $\sigma$ using transpositions $(i\ j)$ with weights $\varphi(i,j)$.
 For any non-negative
weight function $\varphi$, we have
\[
(1/2)D(\pi,\sigma)
\le \dist_\varphi(\pi,\sigma)
\le 2D(\pi,\sigma)
\]
where \[D(\pi,\sigma) = \sum_{i=1}^n {\rm{weight}} \left(p^*(\pi^{-1}(i),\sigma^{-1}(i))\right).\]
\end{thm}
Due to space limitations, the proof is omitted.

\begin{prop}
Let $
\pi'= \arg\min_\pi \sum_{l=1}^m D(\pi,\sigma_i)$ and $\pi^*= \arg\min_\pi \sum_{l=1}^m \dist_\varphi(\pi,\sigma_i)$.
The permutation $\pi'$ is a 4-approximation to the optimal rank aggregate $\pi^*$. If $\varphi$ is a metric, or if it corresponds to weighted Kendall distance, then $\pi'$ is a 2-approximation of $\pi^*$.
\end{prop}
\begin{proof} The first part of the proof follows from the simple observation that
\[\sum_{l=1}^m \dist_\varphi(\pi^*,\sigma_i) \ge (1/2)\sum_{l=1}^m D(\pi^*,\sigma_i)\]
and
\[\sum_{l=1}^m \dist_\varphi(\pi',\sigma_i) \le 2\sum_{l=1}^m D(\pi',\sigma_i).\]
So, by optimality of $\pi'$ with respect to $D$,
\[\sum_{l=1}^m \dist_\varphi(\pi',\sigma_i) \le 4\sum_{l=1}^m \dist_\varphi(\pi^*,\sigma_i)\] and thus $\pi'$ provides a 4-approximation.
The other claim may be proved similarly, by referring to the results of~\cite{farnoud2012sorting}.
\end{proof}

The permutation $\pi'$ can be obtained using minimum weight bipartite matching by letting
\[f(i,j) = {\rm{weight}} \left(p^*(i,j)\right).\]
In particular, for weighted Kendall distance, we let
\[f(i,j)=\sum_{l=i}^{j-1}\varphi(l,l+1).\]

A simple approach for improving the performance of matching based algorithms is to couple them with local descent methods. More specifically, the local descent method works as follows.
Assume that an estimate of the aggregate at step $\ell$ equals $\pi^{(\ell)}$. Let $\Theta_a=\{{(k \,k+1): k=1,\ldots,n-1\}}$ be the set of all adjacent transpositions. Then
\[ \pi^{(\ell+1)}= \pi^{(\ell)} \, \arg\min_{\tau \in\Theta_a} \sum_{i=1}^{m}\dist(\pi^{(\ell)}\,\tau,\sigma_{i}).\]
The search terminates when the cumulative distance of the aggregate from the set $\Sigma$ cannot be decreased further. We choose the starting point $\pi^{(0)}$ to be the ranking $\pi'$ obtained by the minimum weight bipartite matching algorithm. This method will henceforth be referred to as Bipartite Matching with Local Search (BMLS).

\subsection{Vote Aggregation using PageRank}

For a ranking $\pi\in\mathbb{S}_{n}$ and $a,b\in[n]$, $\pi$ is
said to \emph{rank} $a$ \emph{before} $b$ if $\pi^{-1}(a)<\pi^{-1}(b)$.
We denote this relationship with $a<_{\pi}b$. The notation $a\leq_{\pi}b$ has a
similar meaning, and is used in the case that one allows $b=a$.
%Two rankings $\pi$ and
%$\sigma$ \emph{agree} on a pair $\{a,b\}$ of elements if both rank
%$a$ before $b$ or both rank $b$ before $a$. Furthermore, the two
%rankings $\pi$ and $\sigma$ \emph{disagree} on the pair $\{a,b\}$
%if one ranks $a$ before $b$ and the other ranks $b$ before $a$.
%For example, consider $\pi=1234$ and $\sigma=4213$. We have that
%$4<_{\sigma}1$.
% and that $\pi$ and $\sigma$ agree for $\{2,3\}$
%but disagree for $\{1,2\}$.

An algorithm for rank aggregation based on PageRank and HITS algorithms for ranking web pages was proposed in~\cite{dwork2001rank}. PageRank
is one of the most important algorithms developed for search engines used by Google, with the aim of scoring web-pages based on their relevance. Each webpage that has hyperlinks to other
webpages is considered as a voter, while the voter's preferences for candidates is expressed via the hyperlinks. When a hyperlink to a webpage is not present, it is assumed that the voter does not support
the given candidate's webpage.
%Although the exact implementation details of PageRank are not known, it is widely assumed that the graph of webpages is
%endowed with near-uniform transition probabilities.
The ranking of the webpages is obtained by computing the equilibrium distribution of the chain, and ordering the pages according to the values of
their probabilities. The connectivity of the Markov chain provides transitive information about pairwise candidate preferences, and states with high input probability correspond to candidates ranked
highly in a large number of lists.
%Furthermore, the computational complexity of the method is polynomial, $O(n^2m)$.

This idea can be easily adapted to the rank aggregation scenario in several different settings. In such an adaptation,
it is assumed that the states of the Markov chain correspond to the candidates to be voted on
and that the transition probabilities are functions of the votes. Dwork et al.~\cite{dwork2001rank,dwork2001rank-web} proposed four different ways for computing the transition probabilities from the votes.
For completeness, we briefly describe the methods below before we proceed to describe a new approach for evaluating the transition probabilities for the case of the weighted Kendall  distance.

Let $ P $ denote the state transition probability matrix of the chain, with $ P_{ij} $ denoting the probability of going from state (candidate) $i$ to state $j$. Furthermore, let
\[
\alpha_{ij}(\sigma)=\begin{cases}
1, & \qquad\mbox{if }j\le_{\sigma}i,\\
0, & \qquad\mbox{otherwise }
\end{cases}
\] and \[\alpha_{ij}=\sum_{\sigma\in \Sigma} \alpha_{ij}(\sigma).\] That is, $\alpha_{ij}$ is the number of voters that ranked candidate $j$ at least as high as candidate $i$.

In the first case (Case 1), the transition probabilities are computed according to
\[P_{ij}=\frac{I(\alpha_{ij}>0)}{\sum_k I(\alpha_{ik}>0)},\]
where $I(x>0)$ equals 1 if $x>0$ and equals 0 otherwise.
In the second scenario (Case 2) the probabilities are set to
 \[P_{ij} = \frac1m \sum_{\sigma \in \Sigma} P_{ij}(\sigma)\] with
\[P_{ij}(\sigma) = \frac{\alpha_{ij}(\sigma)}{\sum_k \alpha_{ik}(\sigma)}.\]

For Case 3, the transition probabilities are evaluated as
 \[P_{ij} = \frac1m \sum_{\sigma \in \Sigma} P_{ij}(\sigma)\] with
$P_{ij}(\sigma) = \frac{\alpha_{ij}(\sigma)}{n}$ for $j<_{\sigma}i$  and
$P_{ii}(\sigma) = 1-\frac{\sum_j\alpha_{ij}(\sigma)}{n}$.
The fourth method (Case 4) differs from the previous methods in so far that it uses transition probabilities based on majority votes, and will not be used in our subsequent studies.

Our Markov chain model for weighted Kendall distance is similar to Case 3, with a major modification that includes incorporating transposition weights into the transition probabilities.
To accomplish this task, we proceed as follows.

Let $w_k=\varphi(k,k+1)$, and let $i_\sigma = \sigma^{-1}(i)$ for candidate $i,i=1,\cdots,n$. Note that $i_\sigma>j_\sigma$ if and only if $i>_\sigma j$. For $l>k$, let \[w(k:l)=\sum_{h=k}^{l-1}w_h\] denote the sum of the weights of transpositions $(k\ k+1),(k+1\ k+2),\cdots,(l-1\ l)$. We set
\begin{equation}
\beta_{ij}(\sigma) = \max_{l:j_\sigma\le l < i_\sigma} \frac{w(l:i_\sigma)}{ i_\sigma-l }
\label{eq:beta}
\end{equation}
if $j_\sigma<i_\sigma$, $\beta_{ij}(\sigma)=0$ if $j_\sigma>i_\sigma$, and
\[\beta_{ii}(\sigma) = \sum_{k:k_\sigma>i_\sigma}\beta_{ki}(\sigma).\]

%As for the case of uniform wight functions, there are multiple ways of assigning transition probabilities to the candidate graph. We propose two approaches, which we discriminate via the superscripts $1,2$ in the transition probability matrices. For the first variant, we
The transition probabilities equal
\[P_{ij} = \frac1m\sum_{k=1}^m P_{ij}(\sigma_k),\] with
\[P_{ij}(\sigma) = \frac{\beta_{ij}(\sigma)}{\sum_k \beta_{ik}(\sigma)}.\]
% For the second variant, we set \[P^{(2)}_{ij} = \frac{\beta_{ij}}{\sum_k \beta_{ik}}\] if $i>_\sigma j$, $P^{(2)}_{ij}(\sigma) = 0$ if $i<_\sigma j$, and \[P^{(2)}_{ii}(\sigma) = \frac{\sum_{j:j>_\sigma i}\beta_{ij}(\sigma)}{\sum_j \beta_{ij}(\sigma)}.\]

Intuitively, the transition probabilities described above may be interpreted in the following manner.
The transition probabilities are obtained by averaging the transitions probabilities corresponding to individual votes $\sigma\in\Sigma$.
For each vote $\sigma$, let us first consider the case $ j_\sigma = i_\sigma -1$. In this case, the probability of going from candidate $i$
to candidate $j$ is proportional to $w_{j_\sigma}=\varphi(j_\sigma,i_\sigma)$. This implies that if $w_{j_\sigma}>0$, one moves from candidate $i$ to candidate $j$ with positive
probability. Furthermore, larger values for $w_{j_\sigma}$ result in higher probabilities for moving from $i$ to $j$.

Next, consider a candidate $k$ with $ k_\sigma = i_\sigma -2$. In this case, it seems reasonable to let the probability of transitioning from candidate $i$ to candidate $k$ be proportional to
$\frac{w_{j_\sigma}+w_{k_\sigma}}{2}$. However, since $k$ is ranked before $j$ by vote $\sigma$, it is natural to require that the probability of moving to candidate $k$ from candidate $i$
is at least as high as the probability of moving to candidate $j$ from candidate $i$. This reasoning leads to $\beta_{ik}=\max\{w_{j_\sigma},\frac{w_{j_\sigma}+w_{k_\sigma}}{2}\}$ and motivates using
the maximum in \eqref{eq:beta}. Finally, the probability of staying with candidate $i$ is proportional to the sum of the $\beta$'s from candidates placed below candidate $i$.
%It is easy to verify that with uniform weight the above transition probabilities are close to those obtained from Case 3 of \cite{dwork2001rank}.
%The weights $\beta$ represent a measure of the weight of swapping the elements in position $i$ and position $j$, under the weighted Kendall  model. The higher this weight is, the less likely it would be to perform such a swap, so that if a candidate $j$ is ranked higher than $i$, and the weight of swapping $i$ and $j$ is high, then the transition probability from $i$ to $j$ should be high.

\section{Results} \label{sec:results}

The performance of the Markov chain approaches described above cannot be evaluated analytically.
A common approach when dealing with heuristic methods for hard combinatorial optimization problems is to test the
performance of the scheme on examples for which the optimal solutions are easy to evaluate in closed form.

In what follows, we evaluate the various aggregation approaches on a simple test example, with
$m=11$. The set of votes (rankings) $\Sigma$ is given in matrix form by
\[
\left(\begin{array}{ccccccccccc}
1 & 1 & 1 & 2 & 2 & 3 & 3 & 4 & 4 & 5 & 5\\
2 & 2 & 2 & 3 & 3 & 2 & 2 & 2 & 2 & 2 & 2\\
3 & 3 & 3 & 4 & 4 & 4 & 4 & 5 & 5 & 3 & 3\\
4 & 4 & 4 & 5 & 5 & 5 & 5 & 3 & 3 & 4 & 4\\
5 & 5 & 5 & 1 & 1 & 1 & 1 & 1 & 1 & 1 & 1
\end{array}\right).
\]
Here, each column corresponds to a vote, e.g., $\sigma_{1}=\left[1,2,3,4,5\right]$.
Let us consider candidates 1 and 2. Using a plurality rule, one would arrive
at the conclusion that candidate 1 should be the winner, given that 1 appears most often at the top of the list.
Under a number of other aggregation rules, including Kendall's $\tau$ and Borda's method,  candidate 2 would
be the winner.
%The methods to be tested include exhaustive search for finding the OPTimum solution (OPT), Bipartite Matching with Local Search
%(BMLS), and the Markov Chain method (MC).

\begin{center}
\begin{table*}
\begin{centering}
\begin{tabular}{|c|c|c|c|c|}
\hline
\multirow{2}{*}{Method} & \multicolumn{4}{c|}{Aggregate ranking and average distance}\tabularnewline
\cline{2-5}
 & $w=\left[1,0,0,0\right]$ & $w=[1,1,1,1]$& $w=\left[1,1,0,0\right]$ & $w=\left[0,1,0,0\right]$ \tabularnewline
\hline
\hline
OPT & $\left[\underline{1},4,3,2,5\right]$, 0.7273 & $\left[2,3,4,5,1\right],$ 2.3636 & $\left[\underline{2,3},4,5,1\right]$, 1.455 & $\left[\underline{3,2},5,4,1\right]$, 0.636 \tabularnewline
\hline
BMLS & $\left[\underline{1},2,3,4,5\right]$, 0.7273 & $\left[2,3,4,5,1\right],$ 2.3636 & $\left[\underline{2,3},1,5,4\right]$, 1.455 & $\left[\underline{2,3},1,5,4\right]$, 0.636 \tabularnewline
\hline
MC & $\left[\underline{1},2,5,4,3\right]$, 0.7273 & $\left[2,3,4,5,1\right],$ 2.3636 & $\left[\underline{2,1},3,4,5\right]$, 1.546& $\left[\underline{2,3},1,4,5\right]$, 0.636 \tabularnewline
\hline
\end{tabular}
\par\end{centering}
\centering{}\caption{The aggregate rankings and the average distance of the aggregate ranking
from the votes for different weight functions $w$.}\label{tab1}
\end{table*}
\par\end{center}

\begin{center}
\begin{table*}
\begin{centering}
\begin{tabular}{|c|c|c|c|}
\hline
Group & Method & Aggregate Ranking & Average Distance\tabularnewline
\hline
\hline
\multirow{2}{*}{Graduate (28)} & BMLS & \textbf{10}, 12, 9, 8, 1, 3, 2, 11, 7, \textbf{4}, 6, 5 & 5.0918\tabularnewline
\cline{2-4}
& MC & \textbf{10},12, 9, 8, 1, 11, 3, 2, 7, 5, \textbf{4}, 6 & 5.1087\tabularnewline
\hline
\multirow{2}{*}{Undergrad (73)} & BMLS & 12, 9, 8, 1, 3, \textbf{10}, \textbf{4}, 2, 11, 7, 5, 6 & 5.4044\tabularnewline
\cline{2-4}
& MC & 12, 9, 8, 1, 3, \textbf{10}, \textbf{4}, 7, 2, 11, 5, 6 & 5.4106\tabularnewline
\hline
\end{tabular}
\par\end{centering}
\centering{}\caption{Aggregate rankings for undergraduate and graduate students.}\label{tab:Grad-Undergrad}
\end{table*}
\par\end{center}

\begin{center}
\begin{table*}
\begin{centering}
\begin{tabular}{|c|c|c|c|}
\hline
Group & Method & Aggregate Ranking & Average Distance\tabularnewline
\hline
\hline
\multirow{2}{*}{Female, Undergrad (32)} & BMLS & 12, 9, 1, \textbf{8}, 3, 7, 4, 10, 2, 5, \textbf{11}, 6 & 5.3218\tabularnewline
\cline{2-4}
& MC & 12, 9, 8, 1, 3, 10, 7, 2, 5, 4, 11, 6 & 5.3634 \tabularnewline
\hline
\multirow{2}{*}{Male, Undergrad (31)} & BMLS & 12, 9, \textbf{8}, 3, 1, 10, \textbf{11}, 7, 4, 2, 5, 6 & 5.3457 \tabularnewline
\cline{2-4}
& MC & 12, 9, 8, 10, 1, 3, 11, 2, 7, 4, 5, 6 & 5.421 \tabularnewline
\hline
\multirow{2}{*}{DNI, Undergrad (10)} & BMLS & \textbf{8}, 12, 4, 1, 3, 9, 7, 2, 10, \textbf{11}, 6, 5 & 4.2796\tabularnewline
\cline{2-4}
& MC & 12, 8, 4, 1, 3, 9, 10, 11, 7, 2, 6, 5 & 4.4338\tabularnewline
\hline
\end{tabular}

\par\end{centering}

\centering{}\caption{Aggregate rankings for female and male students.}\label{tab:Undergrads}
\end{table*}
\par\end{center}

Our goal is to see how the distance based rank aggregation algorithms would position these two candidates. The numerical results regarding this example are presented in Table \ref{tab1}. In the tables, OPT refers to the optimum solution which was found by exhaustive search and MC refers to the Markov chain method. Furthermore, minimum weight bipartite matching is obtained using \cite{Cao:2011fk}.

If the weight function is $w^{(1)}=[1,0,0,0]$, the optimal aggregate vote clearly corresponds to the plurality winner.
That is, the winner is the candidate with most
voters ranking him/her as the top candidate. A quick check of Table \ref{tab1} reveals that
all three methods identify the winner correctly. Note that the ranks
of candidates other than candidate 1 obtained by the different methods
are different, however this does not affect the distance between the aggregate
ranking and the votes.

The next weight function that we consider is the uniform weight function,
$w^{(u)}=\left[1,1,1,1\right]$. This weight function corresponds to
the conventional Kendall's $\tau$ distance. As shown in Table \ref{tab1},
all three methods produce $\left[2,3,4,5,1\right],$ and the aggregates
returned by BMLS and MC are optimum.

The weight function $w^{(2)}=\left[1,1,0,0\right]$ corresponds
to \emph{ranking of the top 2} candidates. OPT and BMLS return
$2,3$ as the top two candidates, while both preferring $2$ to $3$. The MC method,
however, returns $2,1$ as the top two candidates, with a preference to $2$ over $1$, and
a suboptimal cumulative distance.
It should
be noted that the the MC method is not designed to only minimize the average distance.
Another important factor in determining the winners via the MC method is
that ``winning against strong candidates makes one strong''. In this example,
candidate 1 beats the strongest candidate, candidate 2, three times,
while candidate 3 beats candidate 2 only twice and this fact seems
to be the reason for the MC algorithm to prefer candidate 1 to candidate 3. Nevertheless,
the equilibrium probabilities of candidates 1 and 3 obtained by the MC method are
very close to each other, as the vector of probabilities is $[\underline{0.137},0.555,\underline{0.132},0.0883,0.0877]$.

The weight function, $w^{(t2)}=[0,1,0,0]$, corresponds to \emph{identifying
the top 2} candidates (it is not important which candidate is the
first and which is the second.) The OPT and BMLS identify $\left\{ 2,3\right\} $
as the top two candidates. The MC method returns the stationary probabilities
$\left[0,1,0,0,0\right]$ which means that candidate 2 is an absorbing
state in the Markov chain. This occurs because candidate 2 is ranked
first or second by all voters. The existence of absorbing
states is a drawback of Markov chain methods.
One solution is to remove
2 from the votes and reapply MC.
%In our example the updated set of
%votes become
%\[
%\Sigma'=\left(\begin{array}{ccccccccccc}
%1 & 1 & 1 & 3 & 3 & 3 & 3 & 4 & 4 & 5 & 5\\
%3 & 3 & 3 & 4 & 4 & 4 & 4 & 5 & 5 & 3 & 3\\
%4 & 4 & 4 & 5 & 5 & 5 & 5 & 3 & 3 & 4 & 4\\
%5 & 5 & 5 & 1 & 1 & 1 & 1 & 1 & 1 & 1 & 1
%\end{array}\right)
%\]
%and the weight function becomes $w^{(t2')}=[1,0,0]$.
The MC method in this case
results in the stationary distribution $\left[p\left(1\right),p\left(3\right),p\left(4\right),p\left(5\right)\right]=\left[0.273,0.364,0.182,0.182\right],$
which gives us the ranking $\left[3,1,4,5\right]$. Together with
the fact that candidate 2 is the strongest candidate, we obtain the
ranking $\left[2,3,1,4,5\right]$.

Equipped with this insight, we now perform an aggregation study on a set of rankings collected from UIUC undergraduate and graduate students,
pertaining to criteria for the quality of academic experience (University Climate Study Data), listed below.
The weight function $w=[w_1,\cdots,w_{n-1}]$ was chosen as $w_i=(3/4)^{i-1},i=1,\cdots,n-1$.
\begin{enumerate}
\item Campus friendliness and inclusiveness
\item Availability of recreational and cultural facilities
\item Quality of classrooms and dorms
\item Extracurricular student groups and activities
\item Geographical proximity to family/partner
\item Commitment of campus to build a diverse community
\item Being able to express one's personal identity freely
\item Being able to make friends on campus
\item Safety and security
\item Availability of financial support/scholarship
\item Availability of personal counseling/academic tutoring
\item Friendliness/quality of faculty/instructors
\end{enumerate}

The results of the vote aggregation are presented in Tables \ref{tab:Grad-Undergrad} and \ref{tab:Undergrads}. In Table~\ref{tab:Undergrads}, a group of 10 students who Did Not Indicate their sex is referred to as DNI.
An interesting finding is that the most important criteria for undergraduate students is the effectiveness and friendliness of instructors.

\textbf{Acknowledgment}:  The authors are grateful to Tzu-Yueh Tseng for helping with
the numerical simulations and to Eitan Yaakobi for useful discussions.
The work was supported by the NSF grants  CCF 0821910, CCF 0809895, and CCF 0939370.

\bibliographystyle{ieeetr}
\bibliography{bib}

\begin{thebibliography}{10}

\bibitem{Sinclair-democracy-1988}
R.~K. Sinclair, {\em Democracy and Participation in Athens}.
\newblock Cambridge University Press, 1988.

\bibitem{mueller2008public}
C.~K. Rowley, F.~G. Schneider, and D.~C. Mueller, ``Public choice: An
  introduction,'' in {\em Readings in Public Choice and Constitutional
  Political Economy}, pp.~31--46, Springer US, 2008.

\bibitem{deza2009encyclopedia}
M.~Deza and E.~Deza, {\em Encyclopedia of distances}.
\newblock Springer Verlag, 2009.

\bibitem{arrow1963social}
K.~J. Arrow, {\em Social choice and individual values}.
\newblock No.~12, Yale Univ Pr, 1963.

\bibitem{diaconis1988group}
P.~Diaconis, ``Group representations in probability and statistics,'' {\em
  Lecture Notes-Monograph Series}, vol.~11, 1988.

\bibitem{kumar2010gdr}
R.~Kumar and S.~Vassilvitskii, ``Generalized distances between rankings,'' in
  {\em Proceedings of the 19th international conference on World wide web}, WWW
  '10, (New York, NY, USA), pp.~571--580, ACM, 2010.

\bibitem{farnoud2012sorting}
F.~Farnoud and O.~Milenkovic, ``Sorting of permutations by cost-constrained
  transpositions,'' {\em Information Theory, IEEE Transactions on}, vol.~58,
  pp.~3 --23, Jan. 2012.

\bibitem{cdc-farnoud}
F.~Farnoud, B.~Touri, and O.~Milenkovic, ``{Novel distance measures for rank
  aggregation},'' {\em submitted, CDC 2012}.

\bibitem{kemeney1959mathematics}
J.~G. Kemeny, ``Mathematics without numbers,'' {\em Daedalus}, vol.~88, no.~4,
  pp.~pp. 577--591, 1959.

\bibitem{cook1985ordinal}
W.~D. Cook and M.~Kress, ``Ordinal ranking with intensity of preference,'' {\em
  Management Science}, vol.~31, pp.~26--32, 01 1985.

\bibitem{sculley2007rank}
D.~Sculley, ``Rank aggregation for similar items,'' in {\em Proceedings of the
  Seventh SIAM International Conference on Data Mining}, 2007.

\bibitem{schalekamp2009rank}
F.~Schalekamp and A.~van Zuylen, ``Rank aggregation: Together we're strong,''
  {\em Proc. of 11th ALENEX}, pp.~38--51, 2009.

\bibitem{dwork2001rank}
C.~Dwork, R.~Kumar, M.~Naor, and D.~Sivakumar, ``Rank aggregation revisited,''
  {\em Manuscript (Available at:
  www.eecs.harvard.edu/~michaelm/CS222/rank2.pdf)}, 2001.

\bibitem{dwork2001rank-web}
C.~Dwork, R.~Kumar, M.~Naor, and D.~Sivakumar, ``Rank aggregation methods for
  the web,'' in {\em Proceedings of the 10th international conference on World
  Wide Web}, pp.~613--622, ACM, 2001.

\bibitem{Cao:2011fk}
Y.~Cao, ``An efficient implementation of the munkres algorithm for the
  assignment problem (matlab code).''
  http://www.mathworks.com/matlabcentral/fileexchange/20328, Version 2.3,
  September 2011.

\end{thebibliography}

\end{document}